\documentclass[11pt]{article}
\usepackage{commands}
\usepackage[normalem]{ulem}

\date{}
\begin{document}
\title{\textbf{Average-Radius List-Decodability of \\Random Linear Codes}}
\author{
Venkatesan Guruswami\thanks{{Departments of EECS and Mathematics, and the Simons Institute for the Theory of Computing, UC Berkeley, Berkeley, CA, 94709, USA. Email: \url{venkatg@berkeley.edu}. Research supported by a Simons Investigator Award and NSF grant CCF-2211972.}}
\and 
Shilun Li\thanks{Department of Mathematics, UC Berkeley, Berkeley, CA, 94709, USA. Supported in part by DOD Advanced Research Projects Agency grant HR0011262E031.  Email: \url{shilun@berkeley.edu}.} 
\and
Mihir Singhal\thanks{Departments of EECS, UC Berkeley, Berkeley, CA, 94709, USA. E-mail: \texttt{mihirs@berkeley.edu}. This author was supported in part by an NSF GRFP Fellowship and V.G's Simons Investigator award.} 
}
\maketitle
\thispagestyle{empty}
\begin{abstract}
We prove that for every prime power $q$ and every $p \in (0, 1-1/q)$, a random $\mathbb{F}_q$-linear code of rate $1 - h_q(p) - \epsilon$ is $(p, C_{p,q}/\epsilon)$-average-radius list-decodable with probability at least $1 - q^{-\Omega(n)}$, i.e., for every center $y \in \mathbb{F}_q^n$, the $C_{p,q}/\epsilon$ codewords closest to $y$ have average fractional Hamming distance at least $p$ from $y$. This extends a similar result for (standard) list-decoding due to Guruswami, H\r{a}stad, and Kopparty (2010) to the stronger average-radius guarantee, with the same $O(1/\epsilon)$ list size. For average-radius list-decoding, such a result was previously known only for binary linear codes (Guruswami, Li, Mosheiff, Resch, Silas, and Wootters, 2021) and for general (non-linear) random codes over arbitrary alphabets (Elias, 1991). 
\end{abstract}

\section{Introduction}

This work concerns the (combinatorial) list-decodability properties of random linear codes, for the stronger average-radius variant, at rates close to capacity. We describe the context and prior works, before stating our main result, which is \Cref{thm:main}.

\paragraph{List decoding and random linear codes.}
List decoding, introduced independently by Elias~\cite{elias1957list} and Wozencraft~\cite{wozencraft1958list} in the late 1950s, relaxes the classical requirement of unique decoding by allowing the decoder to output a short list of candidate codewords, one of which is guaranteed to be the transmitted word. A code $\mathcal{C} \subseteq \mathbb{F}_q^n$ is \emph{$(p, L)$-list-decodable} if every Hamming ball of radius $pn$ contains at most $L$ codewords. A standard volume-packing argument shows that the rate must satisfy $R \le 1 - \Hq(p) + o(1)$, and this \emph{list-decoding capacity} is in fact achievable: a random code of rate $R = 1 - \Hq(p) - \epsilon$ is $(p, O(1/\epsilon))$-list-decodable with high probability~\cite{elias1991error}. 
However, a general random code has little structure, and so it is important to pursue the goal of achieving list-decoding capacity with more structured codes.

Linear codes---subspaces of $\F_q^n$ admitting succinct generator-matrix representations---form the canonical structured code family, and are widely used as inner codes in concatenated constructions and as algebraic building blocks in explicit constructions. The question of whether \emph{linear} codes match this performance for list-decoding was more challenging and remained open. Zyablov and Pinsker~\cite{zyablov1981list} showed that a random linear code in $\mathbb{F}_q^n$ of rate $1 - \Hq(p) - \epsilon$ is $(p, q^{O(1/\epsilon)})$-list-decodable, but the exponential dependence on $1/\epsilon$ for the list size was much worse than for general codes.   In addition to its coding-theoretic relevance, the list-decodability of random linear codes is also a basic question in probabilistic combinatorics concerning the relationship between two basic structures---Hamming balls and subspaces.
The challenge in analyzing the list-decodability of linear codes is the linear dependencies between codewords, and how that impacts correlations with respect to proximity to some center.

Guruswami, H\r{a}stad, Sudan, and Zuckerman~\cite{guruswami2002combinatorial} subsequently showed the \emph{existence} of capacity-achieving binary linear codes with list size $O(1/\epsilon)$, though without a high-probability guarantee over random binary linear codes.
Later, Li and Wootters~\cite{li2020improved} extended this argument to show that a random binary linear code of rate $1 - h(p) - \epsilon$ is $(p, h(p)/\epsilon + 2)$-list-decodable with high probability for every $p \in (0, 1/2)$. This list size is also within $2$ of the actual value attained (with high probability) by random binary linear codes~\cite{guruswami2021bounds}.
It is also known that the list size needs to be at least $\Omega(\log (1/\eps))$ for an \emph{arbitrary} binary code of rate $1-h(p)-\eps$~\cite{blinovsky1986bounds,guruswami2014combinatorial}; the exponential gap between the lower and upper bounds for the optimal list size persists.\footnote{For average-radius list-decoding, which is the focus of this paper, the gap is only quadratic; see Section~\ref{sec:related-works}.}

For non-binary alphabets, Guruswami, H\r{a}stad, and Kopparty~\cite{guruswami2010list} proved that a random $\mathbb{F}_q$-linear code of the same rate achieves list size $C_{p,q}/\epsilon$ with probability at least $1 - q^{-\Omega(n)}$. The constant $C_{p,q}$ degrades with large $q$ and also for any fixed $q$ as $p \to 1-1/q$, but for any fixed $p,q$ and $\eps \to 0$, the result showed that a list size of $O(1/\eps)$ can be attained, matching the bound for general codes.

There is also a large and influential line of work on list-decodability of linear codes over large alphabets, driven in part by whether Reed-Solomon codes with random evaluation points achieve list-decoding capacity. Here we just mention the result from \cite{Alrabiah_2025} showing that a random rate-$R$ linear code  over an alphabet of size $2^{O(1/\epsilon^2)}$ is $(1 - R - \epsilon, O(1/\epsilon))$-list-decodable with high probability, and point the reader to \cite{Alrabiah_2025} and references therein for further information.

\paragraph{Average-radius list-decoding.}
We return to our focus on the fixed alphabet regime, aimed at  
understanding the behavior of random linear codes of rate approaching $1 - \Hq(p)$. The specific goal in this work concerns a stronger form of list-decoding called \emph{average-radius} list-decoding. A code being $(p,L)$-list-decodable means that for any $L+1$ distinct codewords and any center, the \emph{maximum} Hamming distance between the codewords and the center exceeds $pn$. For average-radius list-decoding we demand that the same condition holds for the \emph{average} Hamming distance between the codewords and any center (which is a stronger requirement, as the average never exceeds the maximum).

Formally, we say that $\mathcal{C} \subseteq \F_q^n$ is \emph{$(p, L)$-average-radius list-decodable} if for every center $y \in \mathbb{F}_q^n$, the $L+1$ nearest codewords to $y$ have \emph{average} distance more than $pn$ from $y$. 

A $(p, L)$-average-radius list-decodable code is trivially also $(p, L)$-list-decodable. In the other direction, a simple argument shows that if $\mathcal{C}$ is $(p,L)$-list-decodable code, then for any $\gamma > 0$, it is also $(p-\gamma,L/\gamma)$-average-radius list-decodable~\cite[Appendix E]{RudraW15}. However, this incurs a penalty in the list size. 

%For instance, applying this to the $(p, O(1/\epsilon))$-list-decoding bound of \cite{guruswami2010list} only yields $(p-\epsilon,\, O(1/\epsilon^2))$-average-radius list-decodability of random $\F_q$-linear codes at rate $1 - \Hq(p) - \epsilon$.\mscomment{i talk about this two paragraphs down, so maybe this last sentence can be removed?}

It turns out that several proofs of list-decodability in fact already imply (or can be readily modified to yield) the stronger form of average-radius list-decoding, with the \emph{same} list-size. In fact, because the sum is nicer to work with than the maximum, it is often more natural (at least in hindsight) to analyze average-radius list-decodability.  For example, the classical Johnson bound relating the minimum distance of the code to its list-decodability also works for the average-radius version~\cite{cheraghchi2013restricted,RudraW15}. The classical argument for list-decodability of random codes can be adapted to the average-radius setting~\cite{guruswami2014combinatorial}, and this work also systematically isolated average-radius list-decodability as a notion to study. The proofs of list-decodability in the high-noise regime, when $p \to 1-1/q$, based on restricted isometry and related ideas from high-dimensional probability directly work with average-radius list-decoding~\cite{cheraghchi2013restricted,wootters2013list,RudraW14}. The analysis of the list-decodability of binary random linear codes in \cite{guruswami2002combinatorial,li2020improved} was extended, with some technical effort, to the average-radius setting in \cite{guruswami2021bounds}. 
% \cradd{Beyond its role in implying standard list-decoding, average-radius list-decodability is also of intrinsic interest: TODO ADD SOMETHING HERE}

However, the aforementioned proof of Guruswami, H\aa stad, and Kopparty~\cite{guruswami2010list} for non-binary random linear codes relied more heavily on the geometry of the Hamming ball.
%\msmargincomment{is this really the obstruction to extending to average-radius?}. \vg{Not completely, since we show it can be extended. Perhaps it should be toned down}
As such, it was not known whether it could be extended to the average-radius setting. 
The best prior known result for this problem, as observed by Rudra and Wootters in \cite{rudra2018average}, is to directly use the aforementioned observation of \cite[Appendix E]{RudraW15} to get a list size of $O(1/\e^2)$ for average-radius list-decoding at rates within $\eps$ of capacity. Rudra and Wootters in the same paper also provide a different argument that achieves a list size quasipolynomial in $1/\e$. The question remained of whether the $O(1/\epsilon)$ list size matching the standard list-decoding result of \cite{guruswami2010list} can be achieved for average-radius list-decoding.
%Indeed, the best known prior bound for this problem was Corollary 3.4 of \cite{rudra2018average}, which gave a list size that is quasipolynomial in $\e$. The question remained open of whether the $O(1/\epsilon)$ list size shown in \cite{guruswami2010list} can be achieved for average-radius list-decoding, as it is for standard list-decoding. 
This is precisely what we accomplish in this paper. Specifically, we analyze the average-radius list-decodability of random linear codes over $\F_q$ in the style of \cite{guruswami2010list}, and prove the following. 

\begin{thm}[Main]\label{thm:main}
Let $q$ be a prime power, $p \in (0, 1-1/q)$, and $\e > 0$ sufficiently small. For all $n$ sufficiently large, a random $\F_q$-linear code of rate $R = 1 - \Hq(p) - \e$ is $(p,\, C_{p,q}/\e)$-average-radius list-decodable with probability at least $1 - q^{-5n}$, where $C_{p,q}$ is a constant depending only on $p$ and $q$.
\end{thm}

Before we present a brief overview of the proof approach in Section~\ref{sec:pf-overview}, we discuss some other related works.

\begin{table}[h]
\centering
\begin{tabular}{cc}
\toprule
Source & List size $L$ \\
\midrule
Uniform random code  & $O(1/\e)$ \\
\cite{zyablov1981list} & $q^{O(1/\e)}$ \\
\cite{guruswami2010list} + \cite{RudraW15} & $O(1/\e^2)$ \\
\cite{rudra2018average} & $q^{C\log^2(1/\e)}$ \\
\cite{guruswami2021bounds} ($q = 2$ only) & $h(p)/\e + 2$ \\
\textbf{This work} & $C_{p,q}/\e$ \\
\bottomrule
\end{tabular}
\caption{Known upper bounds on average-radius $(p, L)$-list-decodability of uniform random codes and random linear codes with rate $R=1 - \Hq(p) - \e$, for fixed alphabet size $q$ and decoding radius $p \in (0, 1-1/q)$, as $\e \to 0$.}
\label{tab:avg-radius-ld}
\end{table}

%by closely inspecting the proof of the Johnson bound, \cite{cheraghchi2013restricted} proved an average-radius version of the Johnson bound which related the distance of the c
%{\color{red}This formulation was systematically isolated by Guruswami, Håstad, Sudan, and Zuckerman~Many arguments establishing the list-decodability of random codes in fact prove this stronger average-radius guarantee.  However, the proof of \cite{guruswami2010list} on list-decodability of random linear codes does not immediately extend to average-radius list-decodability. 

%\vg{Say that many of the list-decoding proofs for random codes also establish average-radius list-decoding. Cite \cite{guruswami2002combinatorial} for systemtcially isolating average-radius version. Johnson bound holds for average-radius list-decoding (cite \cite{cheraghchi2013restricted}??). However the GHK proof does not (immediately) extend. Goal of this work is to extend it. State the main theorem.}

\subsection{Related works}
\label{sec:related-works}

\textbf{Limitations of average-radius list-decoding.}
In \cite{guruswami2014combinatorial}, the authors proved that a binary code of rate $1-h(p)-\eps$ cannot be $(p,o(1/\sqrt{\eps}))$-average-radius list-decodable. Thus, for average-radius list-decoding, we have a lower bound of $\Omega(1/\sqrt{\eps})$ on list-size which is only quadratically worse than the $O(1/\eps)$ upper bound. In contrast, the best lower bound for list-decoding at rates within $\eps$ of list decoding capacity $1-h(p)$ is $\Omega(\log (1/\eps))$, which is exponentially far from the $O(1/\eps)$ upper bound.

\medskip\noindent\textbf{Thresholds for local properties of random linear codes.}
A highly successful line of recent work has placed list-decodability of random linear codes within a broader structural framework. Mosheiff, Resch, Ron-Zewi, Silas, and Wootters~\cite{mosheiff2020ldpc} introduced the notion of \emph{local properties}: properties defined by excluding any member of a permutation-invariant family of small forbidden subsets of $\mathbb{F}_q^n$. They showed that list-decodability, list-recovery, and average-radius list-decodability are local properties, and proved that every such monotone-decreasing local property has a sharp threshold rate for random linear codes. This framework has been used to transfer random-linear-code guarantees to more structured ensembles, including Gallager's LDPC codes~\cite{mosheiff2020ldpc} and random puncturings of low-bias mother codes due to Guruswami and Mosheiff~\cite{guruswami2022punctured}. A large-alphabet variant based on \emph{local coordinate-wise linear} (LCL) properties was later used by Levi, Mosheiff, and Shagrithaya~\cite{levi2025random} to show a local equivalence between random Reed--Solomon codes and random linear codes for such properties, including list-decodability and list-recoverability. While these frameworks guarantee and characterize threshold rates abstractly, explicitly computing the threshold for a concrete property---for example, the optimal list size $L$ for $(p, L)$-average-radius list-decoding at rate $1 - \Hq(p) - \e$---requires additional property-specific analysis.

\iffalse
\paragraph{Our contribution.}
In this paper we analyze the average-radius list-decodability of random linear codes over $\F_q$ in the style of \cite{guruswami2010list}, and prove the following.

\begin{thm}\label{thm:main}
Let $q$ be a prime power, $p \in (0, 1-1/q)$, and $\e > 0$ sufficiently small. For all $n$ sufficiently large, a random $\F_q$-linear code of rate $R = 1 - \Hq(p) - \e$ is $(p,\, C_{p,q}/\e)$-average-radius list-decodable with probability at least $1 - q^{-5n}$, where $C_{p,q}$ is a constant depending only on $p$ and $q$.
\end{thm}

This strengthens GHK's $(p, C_{p,q}/\e)$-list-decodability~\cite{guruswami2010list} to the stronger average-radius guarantee with the same asymptotic list size.
\noindent
This strengthens GHK's $(p, C_{p,q}/\e)$-list-decodability~\cite{guruswami2010list} to the strictly stronger average-radius guarantee with the same list size, and generalizes the corresponding binary result of~\cite{guruswami2021bounds} to all prime powers $q$. The lower bound $L \ge \Omega_p(1/\sqrt{\e})$ of~\cite{guruswami2014combinatorial} (proved for $q = 2$) applies here as well; the precise optimal scaling of $L$ in $\e$ remains open. Methodologically, the GHK proof relies heavily on the geometry of the Hamming ball and appears tailored to the standard worst-case setting; our principal technical contribution is a reworking of its ``Span-Ball" machinery that yields the average-radius guarantee for all prime powers $q$.
\fi 

\subsection{Proof overview}
\label{sec:pf-overview}
We first recall the strategy of Guruswami, H\aa stad, and Kopparty~\cite{guruswami2010list} (henceforth GHK) for \emph{standard} list-decoding, and then isolate the new ingredient needed for the average-radius guarantee. For a fixed center, GHK bound the number of codewords inside a single Hamming ball $\Ball(0,p)$ by combining two facts: their span-ball theorem (any small linearly independent set of ball vectors spans only $O(1)$ times as many ball vectors), and a counting bound on how many such sets a random linear code can contain. To finish, they use an ``increasing-supports'' argument that extracts, from a linearly independent set, a subsequence whose weights grow noticeably. This last step is tailored to the fixed-radius setting, in which every codeword under consideration genuinely lies in the ball. In the average-radius setting the nearest codewords need only be close \emph{on average}, so this geometric step no longer applies; this is the essential obstacle to porting GHK. Our two new ingredients replace the geometric nature of the GHK argument. We use (i) a Jensen/entropy counting bound (\cref{lem:random-code-bound}) that controls the \emph{total deficit} of a small linearly independent set, and (ii) a greedy-basis argument that transfers this bound to the actual nearest codewords. We now sketch the argument in more detail.

As in \cite{guruswami2010list}, we analyze an arbitrary center $x$ by translating and studying
the low-weight elements of $\spn (C,x)$. Two standard ingredients give
the needed control on small linearly independent sets. First, the span-ball
theorem of \cite{guruswami2010list} implies that, with high probability, the span of any small linearly independent set
\[
  S \subseteq \Ball(0,p)\cap\spn (C,x)
\]
contains only $O(|S|)$ vectors of $\Ball(0,p)$. Second, a Jensen/entropy
count shows that no linearly independent set $S$ of size at most
$L=\Theta_{p,q}(1/\e)$ in $\ensuremath{\spn(\Cc,x)}$ can have large total deficit
\[
  \sum_{v\in S}(pn-|v|).
\]

The main point, then, is to transfer this linearly independent-set statement to the actual
nearest codewords, which may be highly dependent. For a fixed $x$, order the
low-weight elements of $\spn (C,x)$ by increasing weight as
$v_1,v_2,\ldots$, and greedily extract a maximal linearly independent
subsequence $v_{j_1},v_{j_2},\ldots$. The span-ball bound forces this
greedy basis to appear regularly throughout the ordered list: after
$r-1$ basis vectors have been chosen, their span can account for only
$O(r)$ vectors in the ball, so the next basis vector must occur within the
next constant-size block. Writing $D = D(p,q)$ for the span-ball constant of \cref{lem:main} (that is, $D = C_q$ from the span-ball theorem, \cref{lem:span-ball}), the first $DL$ nearby vectors are controlled
by the first $L$ greedy basis vectors. If those $DL$ vectors had average
weight below $(p-2\e)n$, then the corresponding greedy basis vectors would
have total deficit exceeding $\e Ln$, contradicting the entropy-counting
lemma. This gives the desired average-radius bound.

\paragraph{Organization.}
Following this introduction, we state the two key lemmas and prove them in the subsequent subsections: the span-ball bound (\cref{lem:main}), the Jensen entropy counting bound (\cref{lem:random-code-bound}), and finally the main theorem (\cref{thm:main}).

\section{Preliminaries}
For conciseness, we will omit floor and ceiling signs and treat quantities as integers when the analysis would be essentially unaffected.

We define $\Hq(r) := r \log_q(q-1) - r \log_q r - (1-r)\log_q(1-r)$ as the $q$-ary entropy function. For the binary case, we simply write $h(r) := h_2(r)$.

We write $|v|$ for the Hamming weight (number of non-zero coordinates) of a vector $v \in \F_q^n$. For $x \in \F_q^n$, we write $\Ball(x, r)$ for the Hamming ball of radius $rn$ around $x$, i.e., $\{y \in \F_q^n : |y-x| \le rn\}$.

It is a well-known fact that the volume of a Hamming ball is bounded by $q^{\Hq(r)n}$:
\begin{lemma}[Volume of Hamming balls, {\cite[Proposition~3.3.3]{guruswami2019essential}}]
\label{lem:hamming-ball-volume}
For every $x \in \F_q^n$ and every $r \in [0,1-1/q]$,
\[
  \abs{\Ball(x,r)}
  \le q^{\Hq(r)n}.
\]
\end{lemma}

We say that a code $\Cc$ is \emph{$(p, L)$-list-decodable} if, for every $x \in \F_q^n$ and all $L+1$ distinct codewords $v_1, \dots, v_{L+1} \in \Cc$, there is some $i$ such that $|v_i - x| > pn$.

We define a code $\Cc$ to be \emph{$(p, L)$-average-radius list-decodable} if, for every $x \in \F_q^n$ and all $L+1$ distinct codewords $v_1, \dots, v_{L+1} \in \Cc$, we have
\[
  \frac{1}{L+1} \sum_{i=1}^{L+1} |v_i - x| > pn.
\]
That is, the average distance of the $L+1$ codewords closest to any center $x$ is more than $pn$. Note that if a code $\Cc$ is $(p, L)$-average-radius list-decodable, then it is also $(p, L)$-list-decodable.

\section{Main Results}

We prove \cref{thm:main} via the following equivalent formulation, which is what our analysis directly establishes.

\begin{thm} \label{thm:main-equiv}
Let $p \in (0, 1-1/q)$ be fixed and let $\e > 0$ be sufficiently small. Let $\Cc$ be a random linear code of rate $R = 1 - \Hq(p)$. Then there is a constant $C'_{p,q}$ depending only on $p$ and $q$ such that for all $n$ large enough, with probability at least $1-q^{-5n}$, $\Cc$ is $(p-2\e, C'_{p,q}/\e)$-average-radius list-decodable.
\end{thm}

\begin{remark} \label{rem:rephrase}
\cref{thm:main-equiv} is equivalent to \cref{thm:main}: since $\Hq$ has derivative $\Theta_{p,q}(1)$ at $p$, replacing $p$ with $p + 2\e$ in \cref{thm:main-equiv} yields a statement of the form of \cref{thm:main} with $\e$ replaced by $\Theta_{p,q}(\e)$. The capacity-rate form (with a rate of $1 - \Hq(p)$ and radius of $p - 2\e$) used in \cref{thm:main-equiv} is more convenient for the analysis below.
\end{remark}

The proof reduces to establishing the following two lemmas. The first asserts that $\Cc$ cannot contain a small linearly independent set whose span captures too many ball elements (even if we add one arbitrary vector to the code).

\begin{lemma} \label{lem:main}
There exists a constant $D = D(p, q)$ such that for all $n$ large enough, with probability at least $1-q^{-6n}$ over the choice of the random linear code $\Cc$, the following holds. For every $x \in \F_q^n$, every $1 \le \ell \le n^{1/3}$, and every linearly independent set $S \subseteq \Ball(0, p)$ of size $\ell$ with $S \subseteq \spn (\Cc , x)$, we have
\[
  \abs{\spn(S) \cap \Ball(0, p)} \le D\ell.
\]
\end{lemma}

The second lemma controls the total weight deficit of any small linearly independent set $S$ (again adding one arbitrary vector to the code). Recall that the \emph{deficit} of a vector $v$ relative to radius $p$ is $pn - |v|$, which is positive precisely when $v$ lies strictly inside $\Ball(0, p)$.

\begin{lemma} \label{lem:random-code-bound}
Fix $p \in (0, 1-1/q)$ and $\e > 0$, and let $\Cc$ be a random linear code of rate $R = 1 - \Hq(p)$. For all $n$ large enough, with probability at least $1-q^{-6n}$ over the choice of the random linear code $\Cc$, the following holds. For every $x \in \F_q^n$ and every linearly independent set $S \subseteq \spn (\Cc , x)$ of size at most $L = 10/(\e \Hq'(p))$,\footnote{Here, $\Hq'$ denotes the derivative of the entropy function $\Hq$.} we have
\[
\frac{1}{L} \sum_{v \in S} (pn - |v|) \le \e n.
\]
\end{lemma}

We remark that \cite{guruswami2021bounds} established average-radius list-decodability of random \emph{binary} linear codes via a potential-function argument specific to the binary alphabet ($q = 2$) that does not appear to generalize to $q > 2$. Here we instead adapt the span-ball machinery of \cite{guruswami2010list} to the average-radius setting for all prime powers $q$.

\subsection{Proof of \texorpdfstring{\cref{lem:main}}{Lemma~\ref{lem:main}}} \label{sec:span-ball-proof}

For a code $\Cc$ and center $x$, let $E(\Cc,x)$ be the event that the asserted property in \cref{lem:main} fails at $x$. If $x$ is uniform in $\F_q^n$ and independent of $\Cc$, then
\begin{equation} \label{eq:random-center}
  \Pr_{\Cc}\bigl[\exists x \in \F_q^n:E(\Cc,x)\bigr]
  \le q^n\Pr_{\Cc,x}\bigl[E(\Cc,x)\bigr].
\end{equation}
Thus, it is enough to show that $\Pr_{\Cc,x}[E(\Cc,x)] \le q^{-7n}$.
We henceforth assume that $x$ is uniformly random. Let $\Cc^* = \spn (\Cc, x)$; note that $\Cc^*$ is a uniformly random linear code of dimension either $k$ or $k+1$ (depending on whether $x \in \Cc$). Indeed, if $\Cc^*$ has dimension $k$, we can add a uniformly random vector to it (and the statement gets even stronger), so we can assume that $\Cc^*$ is a uniformly random linear code of dimension $k+1$.

Our proof uses two black-box results. The first is the span-ball theorem from \cite{guruswami2010list}. We have changed the constant in the exponent from $5$ to $10$, but the proof in \cite{guruswami2010list} is the same, up to changing the constant chosen at the very end of the proof.

\begin{lemma}[Span-ball theorem, {\cite[Theorem 10]{guruswami2010list}}] \label{lem:span-ball}
For every $\eta_0 > 0$ there exists a constant $C_q = \ensuremath{C_q(\eta_0)} > 0$ such that for all $n$ large enough, all $p \in [\eta_0, 1-1/q-\eta_0]$, and $\ell \le n^{1/3}$: if $X_1, \ldots, X_\ell \in \F_q^n$ are i.i.d.\ uniform on $\Ball(0, p)$, then
\[
  \Pr\bigl[\abs{\spn(X_1,\ldots,X_\ell) \cap \Ball(0, p)} > C_q \ell\bigr] \le q^{-10n}.
\]
The constant $C_q$ does not depend on $p$, $\ell$, or $n$.
\end{lemma}

The second tool controls the probability that a fixed linearly independent tuple falls inside a random subspace.

\begin{lemma} \label{lem:incl-prob}
Let $\Cc^* \subseteq \F_q^n$ be a uniformly random linear code of dimension $k+1$ where $k = (1-h_q(p))n$. For any linearly independent $v_1,\ldots,v_\ell \in \F_q^n$,
\[
  \Pr\bigl[v_1,\ldots,v_\ell \in \Cc^*\bigr] \le q^{-\Hq(p)\ell n + O(\ell)}.
\]
\end{lemma}
\begin{proof}
For each $i$, we have
\[
\Pr[v_i \in \Cc^* \mid v_1,\ldots,v_{i-1} \in \Cc^*]
= \frac{q^{k+2-i}-1}{q^{n-i+1}-1} \le q^{k+1-n},
\]
since after quotienting by $\spn(v_1,\ldots,v_{i-1})$, this is the probability that a fixed nonzero vector lies in a uniformly random $(k+2-i)$-dimensional subspace of an $(n-i+1)$-dimensional space. Therefore,
\[
  \Pr[v_1,\ldots,v_\ell \in \Cc^*] \le \bigl(q^{k+1-n}\bigr)^{\ell} = q^{-\Hq(p)\ell n + O(\ell)},
\]
using $n - k - 1 = \Hq(p)n + O(1)$.
\end{proof}

\begin{proof}[Proof of \cref{lem:main}]
Let $D := C_q$ be the constant from the span-ball theorem (\cref{lem:span-ball}).
Fix $\ell \ge 1$ and define the set of \emph{bad tuples}
\[
  F_\ell(p) :=\Bigl\{(v_1,\ldots,v_\ell) \in \Ball(0,p)^\ell \;:\; v_i \text{ linearly independent},\;\abs{\spn(v_1,\ldots,v_\ell) \cap \Ball(0,p)} > D\ell\Bigr\}.
\]
Since every tuple in $F_\ell(p)$ witnesses a span-ball violation, \cref{lem:span-ball} gives
\[
  \abs{F_\ell(p)} \le \abs{\Ball(0,p)}^\ell \cdot q^{-10n} \le q^{\Hq(p)\ell n - 10n},
\]
where we used $\abs{\Ball(0,p)} \le q^{\Hq(p)n}$. Applying \cref{lem:incl-prob} and taking a union bound over all tuples in $F_\ell(p)$, we get
\begin{align*}
  \Pr\bigl[(v_i) \in F_\ell(p) \text{ with } v_i \in \Cc^* \;\forall\, i\bigr]
  &\le q^{\Hq(p)\ell n - 10n } \cdot q^{-\Hq(p)\ell n + O(\ell)} \\
  &= q^{-10n + O(\ell)}.
\end{align*}
Since $\ell \le n^{1/3}$ (so $O(\ell) = o(n)$), this probability is at most $q^{-9n}$ for all $n$ sufficiently large. Taking a union over all $\ell \le n^{1/3}$:
\[
  \Pr_{\Cc,x}[E(\Cc,x)] \le n^{1/3} \cdot q^{-9n} \le q^{-8n}. \qedhere
\]
\end{proof}

\subsection{Proof of \texorpdfstring{\cref{lem:random-code-bound}}{Lemma~\ref{lem:random-code-bound}}}

To prove this lemma, we will use a counting argument to bound the number of sets $S$ with the given property.

For a code $\Cc$ and center $x$, let $E'(\Cc,x)$ be the event that the asserted property in \cref{lem:random-code-bound} fails at $x$. As in \eqref{eq:random-center},
\[
  \Pr_{\Cc}\bigl[\exists x \in \F_q^n:E'(\Cc,x)\bigr]
  \le q^n\Pr_{\Cc,x}\bigl[E'(\Cc,x)\bigr].
\]
We therefore let $x$ be uniformly random and work with $\Cc^* \supseteq \spn(\Cc,x)$, which is a uniformly random linear code of dimension $k+1$. Thus it is enough to show that $\Pr_{\Cc,x}[E'(\Cc,x)] \le q^{-7n}$.

Set $L := 10/(\e\Hq'(p))$.

\begin{proof}[Proof of \cref{lem:random-code-bound}]
For each $\ell \in \{1,\ldots,L\}$, let $\Bb_\ell$ be the event that there exists a linearly independent $\ell$-tuple $(v_1,\ldots,v_\ell) \subseteq \Cc^*$ with $|v_i| < pn$ for all $i$ and
\[
  \sum_{i=1}^{\ell}(pn - |v_i|) > \e L n.
\]
We wish to bound the probability that $\Bb_\ell$ occurs for some $\ell$. (Note that we were able to make the assumption $|v_i| < pn$ for all $i$ since otherwise we can just remove all such $v_i$ from the tuple and get a smaller tuple that also witnesses some $\Bb_{\ell'}$.)

Define $G_\ell$ to be the set of all linearly independent $\ell$-tuples satisfying this deficit condition; by \cref{lem:incl-prob}, we have
\[
  \Pr[\Bb_\ell] \le \abs{G_\ell} \cdot q^{-\Hq(p)\ell n + O(\ell)}.
\]
Since each deficit satisfies $pn - |v_i| \le pn$, the total deficit is at most $\ell pn$. For $\ell \le \e L/p$ the deficit condition $\sum_i (pn - |v_i|) > \e Ln$ is unsatisfiable, so $G_\ell = \emptyset$ and $\Pr[\Bb_\ell] = 0$.

For $\ell > \e L/p$, the deficit condition forces
\[
  \frac{1}{\ell}\sum_{i=1}^\ell \f{|v_i|}{n} < p - \frac{\e L}{\ell}.
\]
By Jensen's inequality on $\Hq$:
\begin{equation} \label{eq:jensen}
  \sum_{i=1}^{\ell} \Hq\!\left(\frac{|v_i|}{n}\right) \le \ell \cdot \Hq\!\left(\frac{1}{\ell}\sum_i \frac{|v_i|}{n}\right) < \ell \cdot \Hq\!\left(p - \frac{\e L}{\ell}\right).
\end{equation}
Let $(r_1, \dots, r_\ell) = (|v_1|, \dots, |v_\ell|)$ be the weight profile of the tuple (where $r_i \le pn$ for all $i$). By \cref{lem:hamming-ball-volume}, the number of tuples (even including the linearly dependent tuples) with this weight profile is at most $q^{(\Hq(r_1/n) + \cdots + \Hq(r_\ell/n))n}$. Thus, summing over all weight profiles satisfying \eqref{eq:jensen} (of which there are trivially at most $(n+1)^\ell$), we get
\begin{equation} \label{eq:count-bound}
  \abs{G_\ell} \le q^{\ell n \cdot \Hq(p - \e L/\ell) + O(\ell \log n)}.
\end{equation}
By concavity of $\Hq$, we have
\begin{align*}
\Hq(p - \e L/\ell)
&\le \Hq(p) - \frac{\e L}{\ell} \Hq'(p) \\
&\le \Hq(p) - \frac{10}{\ell},
\end{align*}
so plugging this into \eqref{eq:count-bound} gives
\[\abs{G_\ell} \le q^{\Hq(p)\ell n - 10n + O(\ell \log n)}.\]
Combining this with the earlier bound on $\Pr[\Bb_\ell]$ gives
\begin{align*}
  \Pr[\Bb_\ell] &\le q^{\Hq(p)\ell n - 10n + O(\ell \log n)} \cdot q^{-\Hq(p)\ell n + O(\ell)} \\
  &= q^{-10n + O(\ell \log n)} \\
  &\le q^{-9n},
\end{align*}
for $n$ large. Taking a union bound over $\ell$, the failure probability $\Pr_{\Cc,x}[E'(\Cc,x)]$ is at most $L \cdot q^{-9n} \le q^{-8n}$, as desired.
\end{proof}

\subsection{Proof of \texorpdfstring{\cref{thm:main-equiv}}{Theorem~\ref{thm:main-equiv}}}

We condition on the simultaneous occurrence of the events of \cref{lem:main,lem:random-code-bound}; by a union bound, both hold with probability at least
$1 - q^{-6n} - q^{-6n} \ge 1 - q^{-5n}$.

We pick $C'_{p,q} := 10D/\Hq'(p)$, where $D$ is the constant from \cref{lem:main} and $L := 10/(\e\Hq'(p))$ is the size bound from \cref{lem:random-code-bound}, so that $C'_{p,q}/\e = DL$. The target list size is thus $DL$ rather than $L$, with the factor $D$ arising from the greedy-basis reduction below.

We wish to show that, for each $x$, the average distance of the closest $DL$ codewords to $x$ is greater than $(p - 2\e)n$. Note that we can instead show that the average distance of the closest $DL$ elements of $\Cc^* = \spn (C,x)$ to $0$ is greater than $(p - 2\e)n$, since for every $v \in \Cc$, we have $v - x \in \Cc^*$, and including additional elements in $\Cc^*$ beyond those in $\Cc - x$ only decreases the average distance of the closest elements to $0$.

Let $v_1, \dots, v_\kappa$ be the $\kappa$ nearest elements of $\Cc^*$ to $0$, ordered by nondecreasing distance so that $|v_1| \le \cdots \le |v_\kappa|$, where $\kappa$ is the largest value satisfying $\kappa \le DL$ and $|v_\kappa| \le pn$. It suffices to show
\begin{equation*}
\frac{1}{DL} \sum_{i=1}^\kappa (pn - |v_i|) < 2\e n,
\end{equation*}
since if $\kappa < DL$, then elements beyond $v_\kappa$ have $|v_i| > pn$, so $pn - |v_i| < 0$, and including them only decreases the sum. We assume $\kappa > D$, since otherwise the left hand side above is at most $p n / L$ which is at most\footnote{We have $p/L \le p\e\Hq'(p)/10 \le \e/10 < \e$, which follows from the bound $p\Hq'(p) \le \Hq(p) \le 1$; the latter holds by the concavity of $\Hq$, since $0 = \Hq(0) \le \Hq(p) - p\Hq'(p)$ implies $p\Hq'(p) \le \Hq(p)$.} $\e n$.

Let $v_{j_1}, \ldots, v_{j_m}$ (with $j_1 < \cdots < j_m$) be a greedily chosen maximal linearly independent subset of $\{v_1, \ldots, v_\kappa\}$, so that $v_i \in \spn\{v_{j_1},\ldots,v_{j_r}\}$ for all $i < j_{r+1}$.

\begin{lemma} \label{lem:greedy}
For each $1 < r \le \lceil \kappa / D \rceil$, we have $j_r \le D(r-1) + 1$.
\end{lemma}
\begin{proof}
Suppose for contradiction that $j_r > D(r-1) + 1$. Then $v_i \in \spn\{v_{j_1},\ldots,v_{j_{r-1}}\}$ for all $i \le D(r-1) + 1$, giving $\abs{\spn\{v_{j_1},\ldots,v_{j_{r-1}}\} \cap \Ball(0, p)} > D(r-1)$, which contradicts \cref{lem:main}.
\end{proof}

We first verify that $m \ge \lceil \kappa/D \rceil$. By maximality, all of $v_1,\ldots,v_\kappa$ lie in $\spn\{v_{j_1},\ldots,v_{j_m}\}$. If $m < \lceil \kappa/D \rceil$, then $\kappa > Dm$, so $\abs{\spn\{v_{j_1},\ldots,v_{j_m}\} \cap \Ball(0,p)} \ge \kappa > Dm$, contradicting \cref{lem:main} (with $\ell = m \ge 1$). Since $\kappa \le DL$, we have $\lceil \kappa/D \rceil \le L$, so \cref{lem:random-code-bound} applies to the independent set $\{v_{j_1},\ldots,v_{j_{\lceil \kappa/D \rceil}}\}$.

By \cref{lem:greedy}, for each $i > D$ we have $i \ge j_{\ceil{i/D}}$ and hence $|v_i| \ge |v_{j_{\ceil{i/D}}}|$ (since the $v_i$ are sorted by distance). Therefore,
\begin{align*}
\frac{1}{DL} \sum_{i=1}^\kappa (pn - |v_i|)
&= \frac{1}{DL} \sum_{i=1}^{D} (pn - |v_i|) + \frac{1}{DL} \sum_{i=D + 1}^\kappa (pn - |v_i|) \\
&\le \frac{p}{L} \cdot n + \frac{1}{DL} \sum_{i=D + 1}^\kappa (pn - |v_{j_{\ceil{i/D}}}|) \\
&\le \frac{p}{L} \cdot n + \frac{1}{DL} \sum_{r=2}^{\ceil{\kappa/D}} D\,(pn - |v_{j_r}|) \\
&\le \frac{p}{L} \cdot n + \frac{1}{L} \sum_{r=1}^{\ceil{\kappa/D}} (pn - |v_{j_r}|) \\
&\le \frac{p}{L} \cdot n + \e n \qquad \text{(by \cref{lem:random-code-bound})} \\
&< 2\e n.
\end{align*}
Here the last step uses $p/L < \e$, as noted in the prior footnote.

\section{Open Problems}

We conclude with some related open problems.

\paragraph{Better/tight list size bounds.}
Our main theorem gives list size $C_{p,q}/\e$ for average-radius list-decoding at rate $1 - \Hq(p) - \e$. For binary random linear codes, Li and Wootters~\cite{li2020improved} showed that standard list-decoding achieves list size $h(p)/\e + 2$; this was extended in~\cite{guruswami2021bounds} to the average-radius setting, achieving list size $\lfloor h(p)/\e \rfloor + 2$ with the same leading constant $h(p)$. 
In contrast, the constant $C_{p,q}$ in the list-size bound $C_{p,q}/\eps$ in both the GHK proof for list-decoding and our result, is much larger, and in particular blows up for $p \to 1-1/q$. 

An outstanding question is whether one can improve the list-size bound to $C/\eps$ for an absolute constant $C$, and possibly even achieve $C \approx h_q(p)$, to match what is known for the binary case; this is not known for (normal) list decoding either. 
We remark that a proof that a random linear code over $\F_q$ of rate $1-h_q(p)-\eps$ is $(p,C/\eps)$-list-decodable, for some absolute constant $C$, would give a unified statement that applies also to the high-noise regime. In particular, it will imply that random $\F_q$-linear codes of rate $\Omega(\gamma^2)$ are $(1-1/q-\gamma, O(1/\gamma^2))$-list-decodable w.h.p., a fact which is only known
via rather different proofs based high-dimensional probability~\cite{cheraghchi2013restricted,wootters2013list,RudraW14}.

For \emph{arbitrary} (not necessarily linear or random) codes, the combinatorial gap between the lower bound $L \ge \Omega_p(1/\sqrt{\eps})$ and the upper bound $L = O(1/\eps)$ from~\cite{guruswami2014combinatorial} for $(p,L)$-average-radius list-decodable codes of rate $1-h(p)-\eps$ remains open (even for the binary case). Determining the correct polynomial dependence for the optimal codes is a separate central question.

\medskip\noindent\textbf{Average-radius list-recovery.}
A recent work of Doron, Mosheiff, Resch, and Ribeiro~\cite{doron2025list} achieved output list sizes \emph{linear} in $1/\e$ for \emph{standard} (non-average-radius) \emph{list-recovery} of random linear codes over small (constant-size) fields, breaking the Zyablov--Pinsker barrier. Obtaining an average-radius version of their result would be interesting. In particular, can
the Jensen-based counting of the present paper---which handles the average-radius structure directly---be combined with an appropriate analog of the span-ball theorem for list-recovery input sets to recover polynomial list sizes in the average-radius setting?

\medskip\noindent\textbf{Extension to the rank metric.}
Guruswami and Resch~\cite{guruswami2018list} proved that random $\F_q$-linear rank-metric codes (a subspace consisting of $n \times n$ matrices over $\F_q$) of rate within $\eps$ of capacity (i.e., rate $(1-p)^2-\eps$) are list-decodable with list-size $O_{p,q}(1/\eps)$ up to radius $p n$ in the rank-metric. 
Their proof adapts the GHK approach to the rank-metric. Can one prove an average-radius version of their list-decoding result?

\section*{AI Disclosure}
Claude Sonnet 4.6 (Anthropic) and ChatGPT Pro 5.5 (OpenAI) were used to assist with proof checking. The authors are responsible for the correctness and originality of all content including references.
\bibliography{main}
\bibliographystyle{alpha}
\end{document}